\documentclass[10pt,letterpaper]{article}
\usepackage{geometry}
\usepackage{authblk}
\usepackage{graphicx}

\usepackage{algorithm}
\usepackage{algorithmic}[1]
\usepackage{amsmath}
\usepackage{amsfonts}
\usepackage{amssymb}
\usepackage{amsthm}
\usepackage{dsfont}
\usepackage{xspace}
\usepackage{booktabs}

\newcommand{\ie}{{\em i.e.\,}}
\newcommand{\eg} {{\em e.g.\,}}
\newcommand{\st}{\emph{s.t.}\,}

\newcommand{\round}{R}

\newcommand{\rmL}{\mathrm{L}}
\newcommand{\rmI}{\mathrm{I}}
\newcommand{\rmM}{\mathrm{M}}
\newcommand{\calA}{{\cal A}}
\newcommand{\calC}{{\cal C}}
\newcommand{\calE}{{\cal E}}

\newcommand{\calG}{{\cal G}}
\newcommand{\calP}{{\cal P}}
\newcommand{\calR}{{\cal R}}
\newcommand{\calS}{{\cal S}}
\newcommand{\calX}{{\cal X}}
\newcommand{\calL}{{\cal L}}
\newcommand{\onlyline}[5]{
\refstepcounter{#1}
\label{#2}
${#3}_{\ref{#2}}$ & #4 & $\longrightarrow$ & #5
\rule[-0.05in]{0in}{0in}
\\}
\newcommand{\firstline}[5]{
\refstepcounter{#1}
\label{#2}
${#3}_{\ref{#2}}$ & #4 & $\longrightarrow$ & #5 \\}

\newcommand{\lastline}[2]{\rule[-0.05in]{0in}{0in}
& #1 && #2 \\}

\newcommand{\substituteline}[5]{
\refstepcounter{#1}
\label{#2}
${#3}_{\ref{#2}}$ & \hspace{-0.2cm} #4 &\hspace{-0.3cm} $\longleftarrow$ &\hspace{-0.3cm} #5
\rule[-0.05in]{0in}{0in}
\\}

\newcommand{\lfirstline}[3]{\firstline{aep}{#1}{\rmL}{#2}{#3}}
\newcommand{\lonlyline}[3]{\onlyline{aep}{#1}{\rmL}{#2}{#3}}
\newcommand{\refl}[1]{\rmL_{\ref{aep:#1}}}

\newcommand{\isubstitute}[3]{\substituteline{init}{#1}{\rmI}{#2}{#3}}
\newcommand{\refi}[1]{\rmI_{\ref{init:#1}}}

\newcommand{\msubstitute}[3]{\substituteline{merge}{#1}{\rmM}{#2}{#3}}
\newcommand{\refm}[1]{\rmM_{\ref{merge:#1}}}

\newtheorem{theorem}{Theorem}
\newtheorem{lemma}{Lemma}

\theoremstyle{definition}

\newcommand{\clr}{\mathtt{cl}}
\newcommand{\reset}{\mathtt{rst}}
\newcommand{\ipair}{\mathit{Illegal\_Pair}}
\newcommand{\topdown}{\mathit{Down\_OK}}
\newcommand{\bottomup}{\mathit{Up\_OK}}
\newcommand{\varparent}{\mathtt{parent}}
\newcommand{\parent}{\mathit{Par}}
\newcommand{\child}{\mathit{Chi}}

\newcommand{\Dist}{\mathit{Dist}}
\newcommand{\dist}{\mathtt{dist}}
\newcommand{\border}{\mathtt{border}}
\newcommand{\far}{\mathtt{far}}
\newcommand{\groupdist}{\mathtt{groupD}}

\newcommand{\mergedist}{\mathtt{mergeD}}
\newcommand{\Mergedist}{\mathit{MergeDist}}
\newcommand{\Target}{\mathit{Target}}
\newcommand{\target}{\mathtt{target}}
\newcommand{\Candidates}{\mathit{Cand}}

\newcommand{\prior}{\mathtt{prior}}
\newcommand{\Prior}{\mathit{Prior}}
\newcommand{\Minprocess}{\mathbf{Min}}
\newcommand{\Share}{\mathbf{Share}}
\newcommand{\Distance}{\mathbf{Distance}}
\newcommand{\stampon}{\mathtt{stampON}}
\newcommand{\Stampone}{\mathit{Stamp}_1}
\newcommand{\stampone}{\mathtt{stamp1}}
\newcommand{\stamptwo}{\mathtt{stamp2}}
\newcommand{\Stampdist}{\mathit{Stamp}_D}
\newcommand{\stampdist}{\mathtt{stampD}}
\newcommand{\Detector}{\mathit{Detector}}

\newcommand{\Saturated}{\mathit{Saturated}}
\newcommand{\Groupok}{\mathit{GrpOK}}
\newcommand{\Groupsok}{\mathit{GrpsOK}}
\newcommand{\Stampok}{\mathit{StampOK}}
\newcommand{\Groupdistok}{\mathit{GrpDistOK}}
\newcommand{\vfind}{v_{\mathit{find}}}
\newcommand{\vfound}{v_{\mathit{found}}}
\newcommand{\numbergroup}{\#_g}
\newcommand{\numberblack}{\#_b}
\newcommand{\numberprior}{\#_p}
\newcommand{\total}{\#}

\newcommand{\diam}{D}
\newcommand{\id}{\mathit{id}}
\newcommand{\ID}{\mathit{ID}}
\newcommand{\bfsenabled}{\BFS\text{--Enabled}}
\newcommand{\penabled}{\calP\text{--Enabled}}
\newcommand{\aenabled}{\calA\text{--Enabled}}
\newcommand{\xenabled}{\calX\text{--Enabled}}
\newcommand{\mode}{\mathtt{mode}}
\newcommand{\Domain}{\mathit{Domain}}
\newcommand{\domain}{\mathtt{domain}}
\newcommand{\initgroup}{\mathtt{initGroup}}
\newcommand{\Initgroup}{\mathit{InitGroup}}
\newcommand{\nfalse}{n_{\mathit{false}}}
\newcommand{\Height}{{\it Height\/}}
\newcommand{\height}{\mathtt{height}}
\newcommand{\BFS}{\mathrm{BFS}}
\newcommand{\dmax}{k}
\newcommand{\textif}{\textbf{if}\ }

\newcommand{\otherwise}{\textbf{otherwise}}

\newcounter{tbl}

\newcommand{\ta}{T_{\calA}}
\newcommand{\ra}{R_{\calA}}
\newcommand{\loopa}{L_{\calA}}
\newcommand{\tp}{T_{\calP}}
\newcommand{\gsub}{G}
\newcommand{\group}{\mathtt{group}}
\newcommand{\Group}{\mathit{Group}}
\newcommand{\groups}{\mathtt{groups}}

\newcommand{\Merging}{\mathit{Merging}}
\newcommand{\merging}{\mathtt{merging}}
\newcommand{\same}[1]{S_{#1}}

\newcommand{\tr}{\mathbf{true}}
\newcommand{\fl}{\mathbf{false}}
\newcommand{\leg}{\calL_{k}}

\newcommand{\legBFS}{\calL_{\mathrm{BFS}}}
\newcommand{\cfinal}{\calC_{\mathrm{fin}}}
\newcommand{\safea}{\calS_{\calA}}
\newcommand{\safep}{\calS_{\calP}}

\newcommand{\czerofour}{\calC_{\{0,4\}}}

\newcommand{\cgoal}{\calC_{\mathrm{goal}}}
\newcommand{\errorA}{\calE}
\newcommand{\errorM}{\calE}
\newcommand{\errA}{E}
\newcommand{\inputA}{I_\calA}
\newcommand{\inputP}{I_\calP}
\newcommand{\outputA}{O_\calA}
\newcommand{\outputP}{O_\calP}
\newcommand{\outputM}{O_{Merge}}
\newcommand{\invar}[1]{\overline{#1}}
\newcommand{\init}{\mathit{Init}}
\newcommand{\merge}{\mathit{Merge}}
\newcommand{\aep}{\mathrm{Loop}(\calA,\errA,\calP)}
\newcommand{\errmerge}{E}
\newcommand{\mei}{\mathrm{Loop}(\merge,\errmerge,\init)}
\newcommand{\rmcopy}{\mathrm{copy}}
\newcommand{\la}{\leftarrow}

\newcommand{\tH}{\textsuperscript{th}\xspace}
\newcommand{\sT}{\textsuperscript{st}\xspace}

\title{A Self-Stabilizing Minimal k-Grouping Algorithm\thanks{This is a revised version of the conference paper \cite{DLMS17},
which appears in the proceedings of the 18th International Conference on Distributed Computing and Networking (ICDCN), ACM, 2017.
This revised version slightly generalize Theorem \ref{theorem:aep}.}}
\date{}

\author[1]{Ajoy K.\, Datta}
\author[1]{Lawrence L.\, Larmore}
\author[2]{Toshimitsu Masuzawa}
\author[1]{Yuichi Sudo}

\affil[1]{Department of Computer Science University of Nevada, Las Vegas}
\affil[2]{Graduate School of Information Science and Technology, Osaka University}

\begin{document}
\maketitle

%

\begin{abstract}
We consider the minimal $\dmax$-grouping problem:
given a graph $G=(V,E)$ and a constant $\dmax$,
partition $G$ into subgraphs of diameter no greater than $\dmax$,
such that the union of any two subgraphs has diameter greater than $\dmax$.
We give a silent self-stabilizing asynchronous distributed algorithm
for this problem in the composite atomicity model of computation,
assuming the network has unique process identifiers. 
Our algorithm works under the \emph{weakly-fair daemon}.
The time complexity (\ie the number of rounds
to reach a legitimate configuration) of our algorithm is
$O\left(\frac{n\diam}{\dmax}\right)$
where $n$ is the number of processes in the network and
$\diam$ is the diameter of the network.
The space complexity of each process is $O((n +\nfalse)\log n)$
where $\nfalse$ is the number of false identifiers, i.e., identifiers
that do not match the identifier of any process,
but which are stored in the local memory of at least one process
at the initial configuration.
Our algorithm guarantees that
the number of groups is at most $2n/\dmax+1$ after convergence.
We also give a novel composition technique
to concatenate a silent algorithm repeatedly, which we
call \emph{loop composition}.
\end{abstract}


\section{Introduction}\label{sec: introduction}
Modern networks or distributed systems generally
consist of numerous computers (or processes).
Therefore, it is important, in some applications,
to partition such a system into
a set of groups, among each of which
processes communicate with each other without much delay.
This problem has been formalized as $\dmax$-clustering.
In the literature, 
the following similar but different two definitions exist
for $\dmax$-clustering:
given a graph $G(V,E)$ and a constant integer $\dmax$,
one is to find a partition of $V$ into $\{V_1,\dots,V_s\}$
such that every subgraph $\gsub(V_i)$ induced by $V_i$
has radius no greater than $\dmax$,
and the other one is to find a partition of $V$ into
$\{V_1,\dots,V_s\}$ such that
every subgraph $\gsub(V_i)$ has diameter no greater than $\dmax$.
In the former case,
it is also required to designate one process of each subgraph $G(V_i)$
as a cluster-head such that all processes of $G(V_i)$
are located within $\dmax$ hops from the cluster head.
In this paper, we call the former \emph{asymmetric $\dmax$-clustering}
and call the latter 
\emph{$\dmax$-grouping}.

\subsection*{Our Contributions}
This paper considers $\dmax$-grouping.
We aim to construct a minimal $\dmax$-grouping,
since finding the minimum $\dmax$-grouping is known to be NP-hard
for $\dmax > 0$~\cite{Deogun97}.
Specifically, a partition of $V$ into $V_1,\dots,V_s$
is said to be a minimal $k$-grouping
if every two distinct subgraphs $G(V_i)$ and $G(V_j)$
are unmergeable, that is,
$G(V_i \cup V_j)$ has diameter greater than $\dmax$.
We give a self-stabilizing minimal $\dmax$-grouping algorithm
for any undirected graph with unique process-identifiers.
The time complexity of our algorithm is $O(n\diam/\dmax)$ rounds,
where $n = |V|$ and $\diam$ is the diameter of the network, while
the space complexity is $O((n+\nfalse)\log n)$ bits per process,
where $\nfalse$ is the number of false identifiers, i.e., identifiers
stored in the memory of processes in the initial configuration
which do not match the identifier of any process.
Our algorithm also guarantees that the number of groups is at most $2n/k+1$.

We introduce a novel composition technique, \emph{loop composition},
to define our algorithm.
To the best of our knowledge,
this is the first composition technique
that enables the same algorithm (called the \emph{base algorithm})
to be executed arbitrarily many
times repeatedly.
Specifically,
every time an execution of the base algorithm terminates,
a new execution of the base algorithm starts
after the values of all output variables
are copied to the corresponding input variables,
unless a specific condition holds.
This composition technique reuses the same variables
of the base algorithm repeatedly.
It helps us design an algorithm
with small space complexity.
Moreover, it also helps us eliminate certain assumptions.
For example, our algorithm repeats at most $n/\dmax$ executions
of its base algorithm, but it does not need to know
even an upper bound of $n$.

\subsection*{Related Work}
Distributed $\dmax$-grouping algorithms 
are given in~$\cite{Fernandess02,Ducourthial10}$.
Fernandess and Malkhi~\cite{Fernandess02}
give a non self-stabilizing $\dmax$-grouping algorithm in a unit disk graph.
Their algorithm does not guarantee minimality,
but does guarantee an $O(\dmax)$-approxi{-\allowbreak}mation, meaning that the number
of groups it produces is within an $O(k)$ factor of the optimum number
of groups.
Its time complexity is $O(n)$, and its space complexity
is $O(k \log n)$ bits per process.
Ducourthial and Khalfallah~\cite{Ducourthial10}
give a self-stabilizing $\dmax$-grouping algorithm
in a unit disk graph.
Their algorithm guarantees minimality of groups.
However, no upper bound on the number of groups is given.
The upper bound of the time complexity is also not given.

Finding the minimum asymmetric
$\dmax$-clustering is also $NP$-hard~\cite{Amis00}.
Two self-stabilizing asymmetric $\dmax$-clustering algorithms
exist in the literature~\cite{Datta12,Datta13};
The algorithm of \cite{Datta12} is $O(\dmax)$-competitive
while that of \cite{Datta13} guarantees
that clusters it produces are minimal
and the number of the clusters is at most $(k+1)/n$.

There are a variety of techniques to compose two or more
self-stabilizing algorithms
in the literature~\cite{Dolev00,Dolev99,Datta00,Herman92,Datta13,Yamauchi09,
Yamauchi10}, such as fair-composition~\cite{Dolev99}
and parallel composition~\cite{Dolev00},
but none of them enables an unbounded number of repetitions
of the same algorithm, such as ours.


\section{Preliminaries}\label{sec: preliminaries}
A connected
undirected network $G=(V,E)$ of $n=|V|$ processes is given where $n\ge 2$.
Each process $v$ has a unique identifier $v.\id$
chosen from a set $\ID$ of non-negative integers.
We assume that $|\ID| \le O(n^c)$ holds for some constant $c$,
thus, a process can store an identifier in $O(\log n)$ space.
By an abuse of notation, we will identify each process with its identifier,
and vice versa, whenever convenient.
We call a member of $\ID$ a \emph{false identifier}
if it does not match the identifier of any process in $V$.

We use the locally shared memory model~\cite{D74}.
A process is modeled by a state machine
and its state is defined by the values of its variables.
A process can read the values of its own
and its neighbors' variables simultaneously,
but can update only its own variables.
An algorithm of each process $v$ is defined to be
a finite set of actions of the following form:
$<label> <guard>\ \longrightarrow \ <statement>$. 
The \emph{label} of each action is used for reference.
The \emph{guard} is a predicate on the variables and identifiers of $v$ and
its neighbors.
The \emph{statement} is an assignment which
updates the state of $v$.
An action can be executed only if it is \emph{enabled}, \ie
its guard evaluates to true, and
a process is \emph{enabled}
if at least one of its actions is enabled.
The evaluation of a guard and the execution of the corresponding statement
are presumed to take place 
in one atomic step,
according to the composite atomicity model~\cite{Dolev00}.

A \emph{configuration} of the network is
a vector consisting of a state for each process.
We denote by $\gamma(v).x$ the value
of variable $x$ of process $v$ in configuration $\gamma$.
Each transition from
a configuration to another, called a {\em step\/} of the algorithm,
is driven by a {\em daemon}.
We assume the \emph{distributed daemon} in this paper;
at each step, the distributed daemon selects one or more enabled processes
to execute an action.
If a selected process has two or more enabled actions,
it executes the action with the smallest label number.
We write $\gamma \mapsto_{\calA} \gamma'$
if configuration $\gamma$ can change to $\gamma'$ by one step of
algorithm $\calA$.
We define an {\em execution\/} of algorithm $\calA$ to be a
sequence of configurations
$\gamma_0,\gamma_1, \cdots$ such that $\gamma_i \mapsto_{\calA} \gamma_{i+1}$
for all $i \ge 0$.
We assume the daemon to be \emph{weakly-fair}, meaning that a
continuously enabled process must be selected eventually.

An execution is \emph{maximal}
if it is infinite, or it terminates at a {\em final\/} configuration,
\ie a configuration at which no process
is enabled.
Let $\calL$ be a predicate on configurations.
We say that a configuration $\gamma$ of $\calA$ is \emph{safe} for $\calL$
if every execution $\gamma_0,\gamma_1,\dots$ of $\calA$ starting from $\gamma$
(\ie $\gamma_0 = \gamma$) always satisfies $\calL$, that is,
$\calL(\gamma_i)$ holds for all $i \ge 0$.
Algorithm $\calA$ is said to be $\emph{self-stabilizing}$ for $\calL$ 
if there exists a set $\calC$ of configurations of $\calA$ 
such that every configuration in $\calC$ is safe
and  every maximal execution $\gamma_0,\gamma_1,\dots$
of $\calA$ reaches a configuration in $\calC$,
\ie $\gamma_i \in \calC$ holds for some $i \ge 0$.
We also say that $\calA$ is {\em silent\/} if every execution of $\calA$
is finite.
Thus, a silent algorithm $\calA$ is self-stabilizing for predicate $\calL$ if and only if
every final configuration satisfies $\calL$.
In the remainder of this paper, we say ``$\gamma \in \calL$"
or ``configuration $\gamma$ in $\calL$" to mean that $\calL(\gamma) = \tr$.

We measure time complexity of an execution in \emph{rounds}~\cite{Dolev00}. 
We say that process $v$ is \emph{neutralized} at step
$\gamma_i \mapsto \gamma_{i+1}$ if $v$ is enabled at $\gamma_i$
and not at $\gamma_{i+1}$.
We define the first round of an execution
$\varrho=\gamma_0,\gamma_1,\dots$
to be the minimum prefix  
$\gamma_0\ldots \gamma_s$
during which every process enabled at $\gamma_0$
executes an action or is neutralized.
The second round of $\varrho$ is defined to be the first round
of the execution $\gamma_s, \gamma_{s+1}, \ldots$, and so forth.
We evaluate the number of rounds of $\rho$,
denoted by $\round(\varrho)$,
as the \emph{time} of $\varrho$.

Let $N_v$ denote the neighbors of a process $v$,
\ie $N_v = \{u \mid \{u,v\}\in E\}$.
The distance $d(u,v)$ between processes $u$ and $v$ is defined to
be the smallest length of any path between them, where length is
defined to be the number of edges. 
Define $N_v^{i}=\{u \in V \mid d(u,v)\le i\}$, the \emph{$i$-neighborhood}
of $v$. Note that $N_v^{1} = N_v \cup \{v\}$. 
Graph $\gsub(V')$ denotes 
the subgraph of $G=(V,E)$ induced by the set of processes $V' \subseteq V$.
Define $d_{V'}$, or $d_{\gsub(V')}$, to be the distance in $\gsub(V')$.
If there is no path in $\gsub(V')$ between $u$ and $v$, we say
$d_{V'}(u,v) = \infty$.
We denote the diameter of $\gsub(V')$ by $\diam(\gsub(V'))$,
defined to be $\max\{d_{V'}(u,v)\mid u,v\in V'\}$.
We write $\diam=\diam(G)$.
(Note that $\diam(G') = \infty$ if $G'$ is not connected.)

By an abuse of notation,
we sometimes regard a predicate on configurations
as the set of configurations throughout the paper.
For example,
we write $\gamma \in \calL_1 \cap \calL_2$
when a configuration $\gamma$
satisfies both predicates $\calL_1$ and $\calL_2$.

\subsection*{Problem Specification}\label{par: problem}
Given an integer $\dmax$,
our goal is to find a {\em minimal\/} partition\footnotemark{}
$\{V_1,\dots,V_s\}$
with diameter-bound $\dmax$,
which means
(i) $V=\bigcup_{i=1}^s V_i$,
(ii) $V_i \cap V_j = \emptyset$ for any $i\ne j$,
(iii) $\diam(\gsub(V_i)) \le \dmax$ for any $i$,
and $\diam(\gsub(V_i \cup V_j)) > \dmax$ holds for $i \ne j$.
We assume that each process $v$ has 
variable $v.\group \in \ID$. 
We define 
predicate $\leg(\gamma)$
as follows: $\leg(\gamma) = \tr$
if and only if, in configuration $\gamma$,
$\{V(i) \neq \emptyset \mid i \in V\}$
is a minimal partition with diameter-bound $\dmax$
where $V(i) = \{u \in V \mid u.\group = i\}$.
Our goal is to develop a silent self-stabilizing
algorithm for $\leg$.
\footnotetext{
This definition of minimality, which is the same as \cite{Ducourthial10},
does not allow
any pair of groups to be mergeable,
but it allows some three or more groups to be mergeable.
Intuitively, it seems quite costly (perhaps requires exponential time),
even in centralized computing, to
find a ``strict minimal'' solution where no set of groups is mergeable.
}

\section{Loop composition}
\label{sec:composition}
In this section, we give a novel composition
technique to develop a new algorithm.
Given two algorithms $\calA$ and $\calP$
and a predicate $\errA$,
this technique generates 
a silent self-stabilizing algorithm $\aep$
for predicate $\calL$,
whose time complexity is
$O(n+\tp + \ra + \loopa \diam)$
rounds.
As we shall see later,
$\tp$ is an upper bound on the number of rounds
of any maximal execution of  $\calP$,
$\loopa$ is an upper bound on the number of iterations
of $\calA$'s executions,
$\ra$ is an upper bound on the total number of rounds of those (iterated) executions in $\calA$.

\subsection{Preliminaries}
Algorithm $\calA$ is the base algorithm of $\aep$,
and is executed repeatedly
during one execution of $\aep$. 
Algorithm $\calA$ assumes that
the network always stays at a configuration
satisfying a specific condition.
Predicate $\errA$ is used to detect that the network
deviates from that condition, and algorithm $\calP$
is invoked to initialize the network
when such a deviation is detected. 
We define $\outputA$ (resp.~$\outputP$) as the set of variables
whose values can be updated by an action of $\calA$ (resp.~$\calP$),
and $\inputA$ (resp.~$\inputP$) as the set of variables
whose value is never updated and only read
by an action of $\calA$ (resp.~$\calP$).
We assume $\outputA \cap \outputP = \emptyset$
and $\inputP = \emptyset$.
The error detecting predicate $\errA(v)$
is evaluated by process $v \in V$,
and its value depends only on
variables of $\inputA \cup \outputP$ of the processes of $N_v^{1}$.
Let $\errorA$ be a predicate on configurations
such that $\errorA(\gamma)$ holds if and only if
$\bigvee_{v\in V}\errA(v)$ holds in configuration $\gamma$.
We assume that algorithm $\calA$ has a \emph{copying variable}
$\invar{x} \in \inputA$
for every variable $x \in \outputA$.
We define $\gamma^\rmcopy$ as the configuration
obtained by replacing the value of $v.\invar{x}$
with the value of $v.x$ for every process $v$ and every
variable $x \in \outputA$ in configuration $\gamma$.
We define predicate $\cgoal(\calA,E)$ as follows:
configuration $\gamma$ satisfies $\cgoal(\calA,E)$
if and only if
$\gamma \in \lnot \errorA$,
$\gamma^\rmcopy = \gamma$,
and no action of $\calA$ is enabled in any process.
In the rest of this section,
we simply denote $\cgoal(\calA,E)$ by $\cgoal$.
We assume that
$\calA$ satisfies the following three requirements:
\begin{description}
 \item[Shiftable Convergence] 
Every maximal execution of $\calA$
that starts from a configuration in $\lnot\errorA$
terminates 
at a configuration $\gamma$
such that $\gamma^\rmcopy \in \lnot\errorA$.
\item[Loop Convergence]
If $\varrho_0, \varrho_1 \dots$ 
is an infinite sequence of maximal executions of $\calA$
where $\varrho_i =\gamma_{i,0},\gamma_{i,1},\dots,\gamma_{i,s_i}$,
$\gamma_{0,0} \in \lnot \errorA$,
and $\gamma_{i+1,0} = \gamma_{i,s_i}^\rmcopy$ for each $i \ge 0$,
then $\gamma_{j,s_j} \in \cgoal$
and
$\round(\varrho_0)+\round(\varrho_1)+\dots\round(\varrho_j) < \ra$
hold for some $j < \loopa$, and
\item[Correctness]
$\gamma \in \cgoal \Rightarrow \gamma \in \calL$
holds for every configuration $\gamma$.
\end{description}
Algorithm $\calP$ is used to initialize a network
when the network stays at a configuration in $\errorA$. 
We assume that every maximal execution of $\calP$
terminates at a configuration in $\lnot \errorA$ within $\tp$ rounds.

\subsection{Algorithm {\large ${\aep}$}}
In this subsection,
we present an algorithm
$\aep$ given $\calA$, $\errA$, and $\calP$
satisfying the above requirements.
Our algorithm is a silent self-stabilizing
for $\calL$ and its time complexity is
$O(n+\tp + \ra + \loopa \diam)$ rounds.

Our strategy to implement $\aep$ is simple:
The network executes $\calA$ repeatedly while
it stays at configurations in $\lnot \errorA$.
When an execution of $\calA$ terminates,
each process $v$ copies the value of $v.x$
to $v.\invar{x}$ for all $x \in \outputA$,
and the network restarts a new execution of $\calA$,
unless it reaches a configuration in $\cgoal$.
The network executes $\calP$ when
it stays in $\errorA$.
The assumptions of $\calA$ and $\calP$
guarantee that
the network eventually reaches
a configuration in $\calL$.

We give the actions of $\aep$ in Table \ref{tbl:loop6}.
In what follows, we denote by $\xenabled(v)$ that some action of
algorithm $\calX$ is enabled at process $v$.

We use the algorithm of \cite{Datta11},
(denoted by BFS in this paper)
as a module to construct a BFS tree ($\rmL_1$).
The tree is used as a communication backbone
for $\aep$.
Algorithm BFS is silent and self-stabilizing
and constructs a BFS tree:
each process $v$ has variable $v.\varparent \in N_v \cup \{\bot\}$
where $\bot$ means \emph{undefined} or \emph{null value},
and every maximal execution of BFS terminates 
at a configuration in $\legBFS$
where $r.\varparent = \bot$
for some $r \in V$
and the set of edges $\{(v,v.\varparent) \mid v \neq r\}$
spans a BFS tree rooted at $r$.
We define $\parent(v) = \{v.\varparent\}$
and $\child(v) = \{u \in V \mid u.\varparent = v\}$. 
We regard $\{\bot\}$ as the empty set,
hence $\parent(r) = \emptyset$
in a configuration satisfying $\legBFS$.
Thus, a predicate such as ``$\forall u \in \parent(r):\dots $''
always holds, and a predicate
such as ``$\exists u \in \parent(r):\dots $''
never holds.
This algorithm terminates within $O(n)$ rounds
and uses $O(\log n)$ space per process.
We discuss the following part
assuming that an execution of BFS already terminated
and the network stays at configurations in $\legBFS$.

The difficulty we face is to detect termination
of each execution of $\calA$ and $\calP$,
and prevent any two processes from performing different executions
at the same time
(\eg executions of $\calA$ and $\calP$
or the $i$\tH execution and $(i+1)$\sT execution of $\calA$).
We overcome it with a \emph{color wave} mechanism.
Specifically, each process $v$ has three variables
$v.\clr \in \{0,1,2,3,4\}$, $v.\mode \in \{A,P\}$
and $v.\reset \in \{0,1\}$.
Variable $v.\clr$ is the {\em color\/} of process $v$.
Colors 0, 1, and 2 mean that the current execution of $\calA$
or $\calP$ may not have terminated yet,
Color 3 means that the execution has already terminated,
and Color 4 means that the network is now in transition to the next execution.
Variable $v.\mode$ indicates
which algorithm is currently executing.
As we shall see later,
all processes must have the same mode
unless some process has color 4.
Process $v$ executes $\calA$ (resp.~$\calP$)
if $\aenabled(v)$ (resp.~$\penabled(v)$)
and every $u \in N_v^1$ satisfies
$u.\mode = A$ (resp.~$u.\mode = P$)
and $u.\clr \neq 4$
($\refl{aenabled}$-$\refl{penabled}$).
A \emph{reset flag}
is used to prevent any process from having color 3 or 4
until the current execution of $\calA$ or $\calP$ terminates.
Process $v$ raises a reset flag (\ie $v.\reset \la 1$)
every time it executes an action of $\calA$
or $\calP$ ($\refl{aenabled}$--$\refl{penabled}$).
The raised flag is propagated from $v$ to the root $r$
through the BFS tree
($\refl{propagate_reset}$),
resulting in changing $r$'s color to 0
unless it is already 3 or 4 ($\refl{color_reset}$).
The reset flags are eventually dropped 
in order from $v$ to $r$ ($\refl{del_reset}$).
Once $r$'s color is reset to 0,
all other processes will also change their colors to 0
because a process changes its color to 0
when its parent has color 0 ($\refl{color_init}$--$\refl{color_init34}$).

Color waves are propagated through the BFS tree (Figure \ref{fig:color-waves}).
Suppose that all processes have color 0 now.
When $r.\reset = 0$,
the root $r$ begins a top-down color-1-wave
changing all processes' colors from $0$ to $1$.
Specifically, each $v \in V$ changes its color from 0 to 1
when its parent has color 1 ($r$ ignores this condition),
all its children have color 0,
and $v.\reset = 0$ ($\refl{down}$).
After the color-1-wave reaches all leaves,
a bottom-up color-2-wave begins from leaves to $r$
changing all processes' color from $1$ to $2$.
Specifically, each $v \in V$ changes its color from 1 to 2
when its parent has color 1 ($r$ ignores this condition),
and all its children have color 2 ($\refl{to2}$).
As we will prove later, 
the current execution of $\calA$ or $\calP$ has already
terminated when $r$ receives color-2-wave
(Lemma \ref{lem:at_root2}).
Thereafter,
$r$ begins a top-down color-3-wave in the same way
as a color-1-wave ($\refl{down}$).
A process changes its color from 3 to 4
if the network should shift to the next execution
of $\calA$.
Specifically, process $v$ with mode $A$ changes its color to 4
if some process has $\invar{x}\neq x$
for some variable $x \in \outputA$,
and it performs $v.\invar{x} \la v.x$
for all variables $x \in \outputA$
for the next execution of $\calA$ ($\refl{to4_A}$). 
Note that this color-4-wave is propagated by simple flooding
(not through the $\BFS$ tree). 
A process with mode $P$ changes its color to 4 without this condition,
and changes its mode to $A$ ($\refl{to4_P}$).
In both cases, a process changes its color to 4
only after all its children change their colors from 2 to 3
($\refl{to4_A}$, $\refl{to4_P}$);
Otherwise, a color-3-wave may stop at the process.
Finally, a bottom-up Color-0-wave moves from the leaves
to $r$ changing all colors from 4 to 0.
Specifically, a process 
changes its color from 4 to 0 
if its parent has color 4 ($r$ ignores this condition),
all its children have color 0,
and every neighboring process has color 0 or 4 ($\refl{to0}$).
The last condition is needed to prevent
the process from executing the next execution of $\calA$
before all its neighbors finish their copy procedure
($v.\invar{x} \la v.x$).
All processes eventually return to color 0,
and the network starts a new execution of $\calA$.

 \begin{table*}[t]
 \refstepcounter{tbl}
  \label{tbl:loop6}
 \center
 \vspace{-0.5cm}
 \begin{tabular}{llll}
  \toprule
  \multicolumn{4}{c}
  {\textbf{Table \ref{tbl:loop6}: $\aep$}} \\
  \midrule
  \multicolumn{4}{l}{\textbf{[Notations]}}\\

  \multicolumn{4}{l}{
  $\ipair = \{(1,3),(1,4),(2,0),(2,1),(2,3),(2,4),(3,0),(3,1),(4,1),(4,2)\}$
  }\\

  \multicolumn{4}{l}{
  $\topdown(v) \equiv (\forall u \in \parent(v): u.\clr = v.\clr+1) \wedge (\forall u \in \child(v): u.\clr = v.\clr)$
  }\\
  
  \multicolumn{4}{l}{
  $\bottomup(v) \equiv (\forall u \in \parent(v): u.\clr = v.\clr) \wedge (\forall u \in \child(v): u.\clr \equiv v.\clr+1 \pmod{5})$
  }\\
  
  \midrule
  \multicolumn{4}{l}{\textbf{[Actions of process $v$]}}\\
  
  \lonlyline{aep:bfs}{$\bfsenabled(v)$}{execute BFS}
  
  \lonlyline{aep:color_reset}
  {$(v.\clr \notin \{0,3,4\})\wedge(v.\reset = 1)$}
  {$v.\clr \la 0$}
  
  \lonlyline{aep:color_init}
  {$(v.\clr \in \{1,2\})\wedge
 (\exists u \in \parent(v): u.\clr = 0)$}
  {$v.\clr \la 0$}
  
  \lonlyline{aep:color_init34}
  {$(v.\clr \in \{3,4\})\wedge
 (\exists u \in \parent(v): u.\clr = 0)$}
  {$v.\clr \la 0$, $v.\reset \la 1$}
  
  \lonlyline{aep:error}
  {$(v.\mode = A)\wedge\errA(v)
\wedge(\forall u \in N_v^{1}:u.\clr \neq 4)$}
  {$v.\mode \la P$, $v.\reset \la 1$}
  \lonlyline{aep:afindsp}
  {$(v.\mode = A) \wedge (v.\clr \neq 4)
 \wedge (\exists u \in N_v: u.\mode = P)$}
  {$v.\mode \la P$, $v.\reset \la 1$}
  
  \lonlyline{aep:aenabled}
  {$\aenabled(v) \wedge
 (\forall u \in N_v^{1}:u.\mode = A
 \wedge u.\clr \neq 4)$}
  {execute $\calA$, $v.\reset \la 1$}
  
  \lonlyline{aep:penabled}
  {$\penabled(v) \wedge
 (\forall u \in N_v^{1}:u.\mode = P
 \wedge u.\clr \neq 4)$}
  {execute $\calP$, $v.\reset \la 1$}
  
  \lonlyline{aep:illegal}
  {$\exists u \in \child(v):(v.\clr,u.\clr) \in \ipair$}
  {$v.\clr \la 0$, $v.\reset \la 1$}

  \lonlyline{aep:propagate_reset}
  {$(v.\reset = 0) \wedge(\exists u \in \child(v): u.\reset = 1)$}
  {$v.\reset \la 1$}
  
  \lfirstline{aep:del_reset}
  {$(v.\reset = 1)  \wedge (\forall u \in \parent(v): u.\reset = 1)$}
  {$v.\reset \la 0$}
  \lastline
  {\hfill $\wedge (\forall u \in \child(v): u.\reset = 0)$}
  {}  
  
  \lonlyline{aep:down}
  {$(v.\clr \in \{0,2\})\wedge \topdown(v) \wedge (v.\reset = 0)$}
  {$v.\clr \la v.\clr+1$}
  
  \lonlyline{aep:to2}
  {$(v.\clr = 1)\wedge \bottomup(v)$}
  {$v.\clr \la 2$}

  \lfirstline{aep:to4_A}
  {$(v.\clr = 3)\wedge (v.\mode = A)\wedge(\forall u\in \child(v):u.\clr \neq 2)$}
  {$v.\clr \la 4$, $v.\invar{x} \la v.x$}
  \lastline
  {\hfill $\wedge (\exists x \in \outputA: v.\invar{x} \neq  v.x \vee \exists u \in N_v: u.\clr = 4)$}
  {\hfill for all $x \in \outputA$}

  \lonlyline{aep:to4_P}
  {$(v.\clr = 3)\wedge (v.\mode = P)\wedge(\forall u\in \child(v):u.\clr \neq 2)$}
  {$v.\clr \la 4$, $v.\mode \la A$}  
  
  \lfirstline{aep:to0}
  {$(v.\clr = 4) \wedge \bottomup(v) \wedge (\forall u \in N_v:u.\clr \in\{0,4\})$}
  {$v.\clr \la 0$}

  \bottomrule
 \end{tabular}
  \vspace{-0.2cm}
 \end{table*}

 \begin{figure}
  \centering
  \includegraphics[width=\hsize,clip]{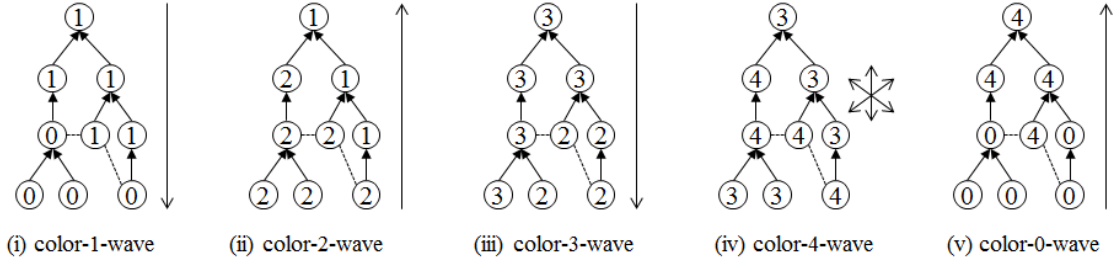}
  \caption{Color-waves in $\aep$.
  Numbers in circles represent colors of processes.
  Arrows represent edges of the BFS tree,
  and dashed lines represent other edges.}
  \label{fig:color-waves}
  \end{figure}

The states of the processes
may be incoherent
in an arbitrary initial configuration,
but $\aep$ resolves the incoherence.
Process $v$ changes
its mode from $A$ to $P$ when $v$ finds
$\errA(v) = \tr$ or 
some neighboring process has mode $P$
($\refl{error}$-$\refl{afindsp}$).
At this time,
$v$ ignores $\errA(v) = \tr$
when some process $u \in N_v^{1}$ has color 4 ($\refl{error}$),
and ignores a neighboring process with mode $P$
when $v.\clr = 4$ ($\refl{afindsp}$).
These exceptions are needed because,
during a color-0-wave,
$\errA(v)$ cannot be evaluated correctly
and a process with mode $A$ can have a neighbor
with mode $P$.
Color-incoherence preventing color waves
is removed as follows:
process $v$ changes its color to 0 if
some $u \in \child(v)$ satisfies
\begin{align*}
 (v.\clr,u.\clr) \in \ipair = &\{(1,3), (1,4), (2,0), (2,1),
 (2,3),\\ &(2,4), (3,0), (3,1), (4,1), (4,2)\}.
\end{align*}
A reset flag is raised every time incoherence is detected and solved
($\refl{color_init34}$, $\refl{error}$, $\refl{afindsp}$, or
$\refl{illegal}$) except for $\refl{color_init}$.
This exception exists only to simplify the proof, as we shall see later.
 
An execution of $\aep$ terminates when
every process $v$
satisfies $v.\mode = A$, $v.\clr =3$,
$v.\reset = 0$,
$\errA(v) = \fl$, $\aenabled(v) = \fl$, and
$v.x = v.\invar{x}$ for all $x \in \outputA$.
We denote the set of such configurations
by $\cfinal$.
Note that $\cfinal \subseteq \cgoal \subseteq \calL$.

\subsection{Correctness}
In this subsection,
we prove that every maximal execution of $\aep$
terminates at a configuration in $\cfinal$
within
$O(n+\tp + \ra + \loopa \diam)$ rounds
regardless of its initial configuration.
For $i = 0,1,2,3,4$, we define $\calR_{i}$ as
the set of configurations in $\legBFS$
where the root $r$ has color $i$.
We define $\calC_S$ as
the set of configurations in $\legBFS$ where
every process has
a color in $S\subseteq\{0,1,2,3,4\}$.

\begin{lemma}
 \label{lem:to0orfinal}
 Let $\varrho = \gamma_0,\gamma_1,\dots$
 be a maximal execution of $\aep$
 starting from $\gamma_0 \in \legBFS$.
 Then, $\varrho$ reaches a configuration in $\calR_{0}$
 or $\cfinal$
 within $O(\diam)$ rounds. 
\end{lemma}
\begin{proof}
 Obtained
 from the following Lemmas \ref{lem:from1},
 \ref{lem:from2}, \ref{lem:from3},
 \ref{lem:from4}, and \ref{lem:from0or4}.
\end{proof}

\begin{lemma}
 \label{lem:from1}
 Let $\varrho = \gamma_0,\gamma_1,\dots$
 be a maximal execution of $\aep$ starting from
 $\gamma_0 \in \calR_{1}$.
 Then, $\varrho$
 reaches a configuration in $\calR_{0}$ or $\calR_{2}$
 within $O(\diam)$ rounds.
\end{lemma}

\begin{proof}
 We assume that no process raises a reset flag
 until $r$'s color changes to 0 or 2.
 This assumption does not lose generality
 since $r$'s color changes to 0 within $O(\diam)$ rounds
 after some process raises a reset flag.
 By $\refl{color_init34}$ and $\refl{illegal}$,
 this assumption guarantees
 $\gamma_0 \in \calC_{\{0,1,2\}}$
 and 
 $\gamma_0(u).\clr = 2 \Rightarrow \gamma_0(v).\clr = 2$
 for all $u,v \in V$ such that $v \in \child(u)$.
 Note that if a process with color 0
 has a child with color 1 or 2,
 the child changes its color to 0
 within one round by $\refl{color_init}$.
 Hence, all processes with color 0
 change their colors to 1 within $O(\diam)$ rounds
 by a color-1-wave ($\refl{down}$),
 and all processes with color 1 including $r$
 change their colors to 2 within the next $O(\diam)$ rounds
 by a color-2-wave ($\refl{to2}$).
\end{proof}

\begin{lemma}
 \label{lem:from2}
 Let $\varrho = \gamma_0,\gamma_1,\dots$
 be a maximal execution of $\aep$ starting from
 $\gamma_0 \in \calR_{2}$.
 Then, $\varrho$
 reaches a configuration in $\calR_{0}$ or $\calR_{3}$
 within $O(1)$ rounds.
\end{lemma}

\begin{proof}
 The root $r$ changes its color to 0
 by $\refl{illegal}$
 if it has a child of color
 0, 1, 3, or 4.
 If all children of $r$ have color 2,
 $r$ changes its color to 3.
 We do not consider a reset flag above,
 but it just changes $r$'s color to 0.
\end{proof}

\begin{lemma}
 \label{lem:2disappear}
 Let $\gamma$ be a configuration in $\legBFS$
 where $u \in V$ has color $3$ and
 $v \in \child(u)$ has color $2$.
 The color of $v$ changes to $0$ or $3$ within $O(1)$ rounds
 in every maximal execution $\varrho$
 of $\aep$ starting from $\gamma$.
\end{lemma}
\begin{proof}
 Process $u$ has color $0$ or $3$
 as long as $v$ has color $2$.
 Hence, $v$ always stays enabled
 unless $v$ executes some action.
 Therefore, $v$ is selected by the scheduler in the first round.
 At this time, $v$ changes its color to $0$
 by $\refl{color_reset}$, $\refl{color_init}$,
 or $\refl{illegal}$, changes its color to $3$
 by $\refl{down}$,
 or raises a reset flag
 by $\refl{error}$, $\refl{afindsp}$,
 $\refl{aenabled}$, $\refl{penabled}$,
 $\refl{propagate_reset}$.
 In the last case, $v$ changes its color to $0$
 by $\refl{color_reset}$ the next time
 the scheduler selects $v$.
\end{proof}

\begin{lemma}
 \label{lem:from3}
 Let $\varrho = \gamma_0,\gamma_1,\dots$
 be a maximal execution of $\aep$ starting from
 $\gamma_0 \in \calR_{3}$.
 Then, $\varrho$ reaches a configuration in
 $\calR_{0}$, $\calR_{4}$, or $\cfinal$
 within $O(\diam)$ rounds. 
\end{lemma}
\begin{proof}
 A process with color 3 changes its color to 0
 by $\refl{illegal}$
 within $O(1)$ rounds if it has a child with
 color 0 or 1.
 A process $v$ with color 3 changes its color to 0 or 4
 by $\refl{illegal}$, $\refl{to4_A}$, or $\refl{to4_P}$
 within $O(1)$ rounds if it has a child with color 4
 because all its children with color 2
 change their colors to 0 or 3 within $O(1)$ rounds
 (Lemma \ref{lem:2disappear}).
 Thus, the network reaches
 a configuration $\gamma_i$ in
 $\calR_0$, $\calR_{4}$, or $\calC_{\{3\}}$
 within $O(\diam)$ rounds.
 When $\gamma_i \in \calC_{\{3\}}\setminus \cfinal$,
 some process's color
 changes to 4 within $O(1)$ rounds,
which changes $r$'s color to 4 within $O(\diam)$
 rounds.
\end{proof}
 
\begin{lemma}
 \label{lem:from4}
 Let $\varrho = \gamma_0,\gamma_1,\dots$
 be a maximal execution of $\aep$ starting from 
 $\gamma_0 \in \calR_{4}$.
 Then, $\varrho$
 reaches a configuration in
 $\calR_{4} \cap \czerofour$
 or $\calR_{0}$ within $O(\diam)$ rounds. 
\end{lemma}

\begin{proof} 
 Lemma \ref{lem:2disappear}
 and $(4,2) \in \ipair$
 guarantee that,
 within $O(D)$ rounds,
 color $2$ disappears from all processes
 as long as $r.\clr = 4$.
 After that,
 every process
 whose parent has color 0 or 4 also has color 0 or 4 within $O(1)$ rounds
 and thereafter never has color 1, 2, or 3
 as long as $r.\clr = 4$.
 Hence, the lemma is proven
 by induction on levels
 (\ie the distance from the root $r$) of processes in the BFS tree. 
 \end{proof}

 \begin{lemma}
 \label{lem:from0or4}
 Let $\varrho = \gamma_0,\gamma_1,\dots$
 be a maximal execution of $\aep$ starting from
 $\gamma_0 \in \calR_{4} \cap \czerofour$.
 Then, $\varrho$
 reaches a configuration in $\calR_{0}$ within $O(\diam)$ rounds.
\end{lemma}
\begin{proof}
 No process changes its color to 1, 2, or 3
 in a configuration of $\calR_{4} \cap \czerofour$.
 Hence, a color-0-wave by $\refl{to0}$
 (or a reset flag)
 changes $r$'s color to 0
 within $O(\diam)$ rounds.
\end{proof}

\begin{lemma}
 \label{lem:at_root2}
 Let $\varrho = \gamma_0,\gamma_1,\dots$
 be a maximal execution of $\aep$
 where $\gamma_0 \in \calR_0$.
 If $i > 0$ is the smallest integer such that
 $\gamma_i \in \calR_2$, 
 then every process $v$
 has the same mode $\calX$ ($\calA$ or $\calP$)
 and satisfies
 $\lnot \xenabled(v)$,
 $v.\clr = 2$ and $v.\reset = 0$
 at configuration $\gamma_i$. 
\end{lemma}

\begin{proof}
  Let $j\ (0 < j < i)$ be the largest integer such that
 $\gamma_{j-1}(r).\clr = 0$ and $\gamma_{j}(r).\clr = 1$.
 A color-1-wave beginning at $\gamma_j$ reaches every leaf,
 and thereafter a color-2-wave is initiated at leaves
 and reaches $r$
 in $\gamma_j,\dots,\gamma_i$.
 No process has color 3 or 4
 in $\gamma_j$; Otherwise, such process $v$ changes its color
 from 3 or 4 to 0 and raises a reset flag
 before $v$ receives a color-1-wave,
 which changes
 $r$'s color to 0 again;
 This is a contradiction to $j$'s definition.
 Since no process changes its color to 3 during
 $\gamma_j,\dots,\gamma_i$, every process has color 0, 1, or 2
 during $\gamma_j,\dots,\gamma_i$.
 In the same way, we can prove that
 all processes have the same mode $\calX$ ($\calA$ or $\calP$),
 and $\lnot \xenabled(v)$ and $v.\reset = 0$ hold
 during $\gamma_j,\dots,\gamma_i$,
 hence no process changes its color from 1 or 2 to 0
 in the mean time.
 Thus, every process $v$
 has the same mode $\calX$ ($\calA$ or $\calP$)
 and satisfies
 $\lnot \xenabled(v)$,
 $v.\clr = 2$ and $v.\reset = 0$
 at configuration $\gamma_i$.
\end{proof}

 \begin{table*}
 \center
 \refstepcounter{tbl}
 \label{tbl:variables}
  \begin{tabular}{r}
   \toprule
   \multicolumn{1}{c}{
   \textbf{Table \ref{tbl:variables}:
   Non-copy Variables used in $\init$ and $\merge$}}\\
   \midrule
   \multicolumn{1}{l}{\hspace{-0.25cm}
   \textbf{[Non-array variables]}}\\
    ${\begin{aligned}
       v.\domain \in 2^{\ID},\\
       v.\height \in [0,\lfloor k/2 \rfloor],
       v.\initgroup,
       v.\group \in {\ID}
   \end{aligned}}$\\ \\
   \multicolumn{1}{l}{\hspace{-0.25cm}   
   \textbf{[Array variables ($u \in v.\domain$)]}}\\
      ${\begin{aligned}
v.\dist[u], v.\groupdist[u], v.\mergedist[u],
v.\stampdist[u] \in [0,2k]\cup\{\bot\},\\ 
v.\border[u], v.\far[u], v.\target[u], v.\stampone[u], 
v.\stamptwo[u], v.\groups[u] \in \ID\cup\{\bot\},\\
v.\merging[u], v.\stampon[u], v.\prior[u] \in \{\fl,\tr\}	  
	\end{aligned}}$\\
   \bottomrule
  \end{tabular}
 \end{table*}

We define $\safep$ to be the set of configurations
in $\calR_{0} \cap \calC_{\{0,1,2\}}$ where
every process has mode $P$,
and define $\safea$ to be the set of configurations
in $\calR_{0} \cap \calC_{\{0,1,2\}}$ and $\lnot \errorA$ such that
every process has mode $A$.

\begin{lemma}
 \label{lem:to_safea_or_safep}
 Let $\varrho = \gamma_0,\gamma_1,\dots$ be
 a maximal execution of $\aep$ starting from 
 $\gamma_0 \in \calR_{0}$.
 Then, $\varrho$ reaches a configuration in $\safea$ or $\safep$
 within $O(\diam)$ rounds.  
\end{lemma}

\begin{proof}
 Since $\gamma_0(r).\clr = 0$,
 every process with color 3 or 4
 changes its color to 0 within $O(\diam)$ rounds.
 At this time, each such process raises a reset flag
 by $\refl{color_init34}$--$\refl{propagate_reset}$,
 which will change $r$'s color to 0.
 Thus, the network reaches a configuration $\gamma$
 in $\calR_{0} \cap \calC_{\{0,1,2\}}$ within $O(\diam)$ rounds.
 If two processes have different modes
 or some process $v$ has $(v.\mode = A) \wedge (\errA(v) = \tr)$
 in $\gamma$,
 then all processes with mode $A$ change their modes
 to $P$ by $\refl{error}$ and $\refl{afindsp}$ within $O(\diam)$ rounds,
 keeping $r$'s color 0
 with a reset flag,
 which proves the lemma.
\end{proof}

\begin{lemma}
 \label{lem:from_safep}
 Let $\varrho= \gamma_0,\gamma_1,\dots$
 be a maximal execution of $\aep$ starting from $\gamma_0 \in \safep$.
 Then, $\varrho$ reaches
 a configuration in $\safea$ within $O(\tp + \diam)$ rounds.
\end{lemma}

\begin{proof}
 Every $\calP$--Enabled process performs $\calP$'s action by
 $\refl{penabled}$ at least once with every one or two rounds
 until $r$'s color changes to $2$.
 Lemma \ref{lem:at_root2} guarantees that
 $r$'s color never changes to $0$
 until an execution of $\calP$ terminates.
 Hence, the network reaches within $O(\tp)$ rounds
 a configuration in $\lnot \errorA$ where
 no process is $\calP$--Enabled.
 Thereafter, all reset flags are dropped within $O(\diam)$ rounds
 by $\refl{del_reset}$,
 and $r$'s color is 0 at the resulting configuration.
 The network reaches a configuration in $\calC_{\{2\}}$
 and again reaches a configuration $\gamma \in \calR_0$
 within the next $O(\diam)$ rounds
 by $\refl{color_init}$,
 $\refl{down}$ (a color-1-wave, a color-3-wave),
 $\refl{to2}$ (a color-2-wave),
 $\refl{to4_P}$ (a color-4-wave),
 and $\refl{to0}$ (a color-0-wave).
 Lemma \ref{lem:at_root2} guarantees $\gamma \in \safea$.
\end{proof}

\begin{lemma}
 \label{lem:from_safea}
 Let $\varrho= \gamma_0,\gamma_1,\dots$ be
 a maximal execution of $\aep$
 starting from $\gamma_0 \in \safea$.
 Then, $\varrho$ reaches
 a configuration in $\cfinal$
 within $O(\ra + \loopa \diam)$ rounds.
%
 \end{lemma}

\begin{proof}
 We can prove, in a similar way as the proof
 of Lemma \ref{lem:from_safep},
 that execution $\varrho$ reaches within $O(\ta+\diam)$ rounds
 a configuration in $C_{\{2\}}$
 where no process is $\calA$--Enabled.
 If $v.\invar{x} = v.x$ holds
 for all $v\in V$ and $x \in \outputA$,
 the network reaches a configuration in $\cfinal$
 by $\refl{down}$ within the next $O(\diam)$ rounds.
 Otherwise, all processes $v$ perform $v.\invar{x} \la v.x$
 by $\refl{to4_A}$,
 and the network returns to a configuration in
 $\safea$ by $\refl{to0}$ within $O(\diam)$ rounds.
 The above execution is repeated
 until reaching a configuration in $\cfinal$,
 and
 the number of iterations is bounded by $\loopa$.
 Since the total number of rounds in those
 iterated executions of $\calA$ is at most
 $\ra$, execution $\rho$ reaches a configuration
 in $\cfinal$
 within $O(\ra+\loopa \diam)$ rounds.
\end{proof}

\begin{theorem}
 \label{theorem:aep}
 $\aep$ is silent and self-stabilizing
 for $\calL$. The execution of $\aep$ 
 terminates within
 $O(n+\tp + \ra + \loopa \diam)$
 rounds.
 The space complexity per process
 is $O(\log n)$ bits.
\end{theorem}

\begin{proof}
 The correctness and the time complexity is obtained from Lemmas 
 \ref{lem:to0orfinal},
 \ref{lem:to_safea_or_safep},
 \ref{lem:from_safep},
 and \ref{lem:from_safea}.
 The variables used by all actions except for $\rmL_1$
 require only constant space per process,
 however, the module BFS ($\rmL_1$) requires $O(\log n)$ bits
 of space per process.
\end{proof}

One may think that
a classical self-stabilizing reset algorithm
(such as the one described in \cite{Dolev00})
can be used to simplify $\aep$
instead of using the five-color waves.
However, a self-stabilizing reset algorithm
cannot be used directly to implement $\aep$.
A reset algorithm aims to reset or initialize the
network to a legal initial configuration (after the illegal
configuration is detected)
while $\aep$ aims to execute the base algorithm $\calA$
repeatedly.
Unlike reset algorithms, a reset signal in Loop is not used to
initialize the network;
a reset signal itself does not trigger
an initialization algorithm $\calP$.
Instead, a reset signal is used just
to change the color of the root to zero (i.e. to initiate the color waves).
Five kinds of color waves described above
guarantee that the $i+1$\tH execution of the base algorithm starts only
after the $i$\tH execution terminates.

\section{Minimal $\dmax$-Grouping Algorithm}
\label{sec:algorithm}

\begin{table*}
 \refstepcounter{tbl}
 \label{tbl:init}
 \center
 \begin{tabular}{cc}
  \toprule
  \multicolumn{2}{c}
  {\textbf{Table \ref{tbl:init}: $\init$}}\\
  \midrule
  \multicolumn{2}{l}{\hspace{-0.25cm} \textbf{[Actions of process $v$]}}\\
  \hspace{-0.5cm}
  \begin{minipage}[c]{0.4\hsize}
  \begin{tabular}{llll}
   \isubstitute{init:domain}
   {$v.\domain$}
   {$\Domain(v)$}
   \isubstitute{init:dist}
   {$v.\dist[u]$}
   {$\Dist(v,u)$}
   \isubstitute{init:height}
   {$v.\height$}
   {$\Height(v)$}
   \isubstitute{init:initgroup}
   {$v.\initgroup$}
   {$\Initgroup(v)$}
   \isubstitute{init:group}
   {$v.\invar{\group}$}
   {$\Initgroup(v)$}   
  \end{tabular}
  \end{minipage}
  &
  \begin{minipage}[c]{0.55\hsize}
  \begin{tabular}{llll}
   \isubstitute{init:groups}
   {$v.\invar{\groups}[u]$}
   {$\Share(v,u,\invar{\groups},v.\invar{\group})$}
   \isubstitute{init:groupdist}
   {$v.\invar{\groupdist}[u]$}
   {$\Distance(v,u,\invar{\groupdist},\same{v})$}
   \isubstitute{init:stampon}
   {$v.\invar{\stampon}[u]$}
   {$\fl$}
   \isubstitute{init:prior}
   {$v.\invar{\prior}[u]$}
   {$\fl$}
  \end{tabular}
  \end{minipage}\\
  \bottomrule
 \end{tabular}
 \vspace{-0.2cm}
\end{table*}

 \begin{table*}
 \center
 \refstepcounter{tbl}
 \label{tbl:initfunctions}
  \begin{tabular}{c}
   \toprule
       \textbf{Table \ref{tbl:initfunctions}: Functions used to describe
   $\init$}\\
   \midrule
    ${\begin{aligned}
 \Dist(v,u) &=
 \Distance(v,u,\dist,\{w \in N_v \mid u \in w.\domain\}),\\
 \Domain(v) &= \{u \mid \Dist(v,u) \neq \bot,\ \Dist(v,u) \le k+1\},\\
 \Height(v) &=
 \begin{cases}
  0 & \textif |\child(v)| = 0\\
  \max \{u.\height + 1 \bmod \lfloor \dmax /2  + 1\rfloor
  \mid u \in \child(v)\}& \otherwise,
 \end{cases}\\
 \Initgroup(v) &=
 \begin{cases}
  v & \textif |\parent(v)|=0
  \vee v.\height = \lfloor \dmax /2 \rfloor \hspace{-7pt}\\
  (v.\varparent).\initgroup & \otherwise.
 \end{cases}\\
      \end{aligned}}$\\
   \bottomrule
  \end{tabular}
 \end{table*}
 
 \begin{table*}[t]
 \center
 \refstepcounter{tbl}
 \label{tbl:error}
  \begin{tabular}{c}
   \toprule
       \textbf{Table \ref{tbl:error}: Error-detecting variable $\errmerge(v)$}\\
   \midrule
    $\begin{aligned}
  \errmerge(v)
 &\equiv \lnot (
 (v.\domain = \Domain(v)) \wedge (\forall u \in v.\domain: v.\dist[u] = \Dist(v,u))\\
 & \hspace{0.8cm}
 \wedge (v.\height = \Height(v))
 \wedge (v.\initgroup = \Initgroup(v))\\
 & \hspace{0.8cm}
 \wedge \Groupok(v) \wedge \Groupsok(v) \wedge \Groupdistok(v)\\
 &\hspace{0.8cm}
 \wedge \forall u \in v.\domain:v.\invar{\stampon}[u]\Rightarrow
 \Stampok(v,u)\\
 &\hspace{0.8cm} \wedge
 (\forall u \in N_v, \forall w \in v.\domain \cap u.\domain:
 v.\invar{\prior}[w]=u.\invar{\prior}[w])),\\
 \Groupok(v) &\equiv
 \forall u \in N_v:
 u.\initgroup = v.\initgroup
 \Rightarrow u.\invar{\group} = v.\invar{\group},\\
 \Groupsok(v) &\equiv
 \forall u \in v.\domain:
 v.\invar{\groups}[u] = \Share(v,u,\invar{\groups},v.\invar{\group})
 \wedge l_v \in g_v(l_v),\\ 
 \Groupdistok(v) &\equiv
 \forall u \in g_v(l_v):
 v.\invar{\groupdist}[u]      
 = \Distance(v,u,\invar{\groupdist},\same{v})
 \neq k+1,\\
  \Stampok(v,u) &\equiv
 \forall w \in S_v: w.\invar{\stampon}[u] \wedge
 \forall w \in N_v(u): w.\invar{\stampon}[l_v]\\
 &\hspace{13pt}
 \wedge
 (s_1 \neq \bot \vee s_2 \neq \bot) \wedge s_D \neq \bot
 \wedge (s_2= \bot \Rightarrow \forall w \in N_v(u):
 w.\invar{\stamptwo}[l_v] \neq \bot)\\
 &\hspace{13pt} 
 \wedge \forall w \in S_v:
 (s_1,s_2) = (w.\invar{\stampone}[u],w.\invar{\stamptwo}[u])\\
 & \hspace{13pt}  \wedge (s_2 = v \Rightarrow s_D=k+1)
 \wedge \forall i \in \{1,2\}: (s_i \neq \bot
 \Rightarrow s_i \in g_v(l_v))\\
 &\hspace{13pt} \wedge (
 (s_1,s_2,s_D) \neq (v,\bot,0) \Rightarrow \\
 &\hspace{1.5cm}
 s_D = 1+\min(\{w.\invar{\stampdist}[u]
 \mid
 w \in S_v\}
 \cup \{w.\invar{\stampdist}[l_v] \mid w \in N_v(u)\}))
       \end{aligned}$\\
   \hspace{-1cm}
   where $(s_1,s_2,s_D) = (v.\invar{\stamptwo}[u], v.\invar{\stamptwo}[u], v.\invar{\stampdist}[u])$.\\
   \bottomrule
  \end{tabular}
 \end{table*}       

In this section, we give a silent self-stabilizing algorithm
for minimal $\dmax$-grouping using
the loop composition method described in the previous section.
In the following, we call a set of processes a \emph{group}.
Two distinct groups $g_1, g_2 \subseteq V$ are said to be
\emph{mergeable} if $\diam(\gsub(g_1 \cup g_2)) \le k$.
Two distinct groups $g_1, g_2 \subseteq V$
are \emph{near} if
$\{v_1,v_2\}\in E$ holds for some $v_1 \in g_1$ and $v_2 \in g_2$
and 
every $u_1 \in g_1$ and $u_2 \in g_2$
are within $\dmax$ hops
in $G$ (not necessarily in $\gsub(g_1 \cup g_2)$). 
Note that two groups are mergeable
only if they are near, 
but two groups are not necessarily mergeable
even if they are near.

\subsection{Overview}
The proposed algorithm $\mei$ consists of
two algorithms $\init$ and $\merge$.
Roughly speaking,
$\init$ makes an initial partition $g_1,\dots,g_s$ 
where $s \le 2n/k + 1$,
and $\merge$ merges two groups
as long as two mergeable groups exist.
Specifically, in an execution of $\merge$,
all processes of each group $g_i$ first agree
on choosing one of $g_i$'s near groups,
say $g_j$, as the target, and
next check whether $g_i$ and $g_j$ are mergeable.
If $g_i$ and $g_j$ are mergeable
and $g_j$ also targets $g_i$,
then they are merged;
if they are not mergeable,
the \emph{stamp} indicating they are not mergeable
is generated on the memories of all the processes of $g_i$ and $g_j$.
This stamp removes $g_j$ from the set of target candidates of $g_i$
in the following executions of $\merge$. 
However, this stamp is not permanent;
it is removed
when $g_i$ and/or $g_j$ merges with another group.
This removal is needed because
the two groups stamped as unmergeable may become mergeable when one of
them is merged to another group. 
After a finite number of executions of $\merge$,
the network reaches a configuration in $\cgoal(\merge,\errmerge)$,
in which all the groups have stamps for all their near groups.
As we shall see later, we assign \emph{prior} labels
to each group so that prior groups
have higher priority of becoming targets than
non-prior groups.
This strategy is to ensure
that the number of executions of $\merge$
is bounded by $O(n/k)$.

\subsection{Preliminaries}
Actions of $\init$ and $\merge$
are given in Tables \ref{tbl:init} and \ref{tbl:merge}
respectively.
By the definition of the composition algorithm,
we can use $v.\varparent$ (thus $\parent(v)$ and $\child(v)$)
obtained from algorithm $\BFS$ in $\init$ and $\merge$,
and assume that the network remains in $\legBFS$
in what follows.
Except for copying variables,
each process $v$ has four non-array variables
and the following thirteen array variables,
as shown in Table \ref{tbl:variables}.
Every copying variable $\invar{x}$ has the same range
as the corresponding variable $x \in \outputM$.
In Tables \ref{tbl:init} and \ref{tbl:merge},
``$v.x \longleftarrow \chi(v)$'' means
``$v.x \neq \chi(v) \longrightarrow v.x \la \chi(v)$'',
and ``$v.x[u] \longleftarrow \chi(v,u)$'' means
``$\exists u \in v.\domain:v.x[u] \neq
f(\chi(v,u)) \longrightarrow v.x[u] \la f(\chi(v,u))$
for all $u \in v.\domain$''
where $f(a) = a$ if $a$ belongs to the range of $v.x[u]$;
otherwise $f(a) = \bot$ for any $a$.
We define $\min \emptyset = \bot$
and $\min S \cup {\bot} = \min S$ for any set $S$.

We denote $l_v = v.\invar{\group}$,
$N_v(u) = \{w \in N_v \mid l_w = u\}$,
$\same{v} = N_v(l_v)$,
$g_v(u) = \{w \in v.\domain \mid v.\invar{\groups}[w] = u\}$,
and $g(u) = \{w \in V \mid w.\invar{\group} = u\}$.
Note that the value of these notations are based on
copying variables $\invar{\group}$ and $\invar{\groups}$,
and the values of $l_v$, $N_v(u)$, $\same{v}$, $g_v(u)$
are computable locally at process $v$
while we use $g(u)$ only in explanations and proofs.

We use three macros $\Share$, $\Minprocess$,
and $\Distance$.
For $v,u \in V$, array variable $x$,
and the function $f'$ on the variables of a process,
we define $\Share(v,u,x,f'(v)) = f'(v)$ if $v = u$;
otherwise $\Share(v,u,x,f'(v)) =
\min \{w.x \mid w \in N_v, v.\invar{\dist}[u]=w.\invar{\dist}[u]
+1\}$.
The action
``$v.x[u] \longleftarrow \Share(v,u,x,f'(v))$''
lets every process $v \in V$
stores $f'(u)$ on $v.x[u]$
for every $u \in N_v^{\dmax+1}$.
For example,
for every $v \in V$ and $u \in N_v^{\dmax+1}$,
$v$ eventually
has $v.\groups[u] = u.\group$ by
action $\refm{groups}$ in Table \ref{tbl:merge}.
For $v \in V$, variable $x \in \ID$,
and predicate $Q$ on variables of a process,
we define $\Minprocess(v,x,Q(v)) = \min\{v,w\}$
if $Q(v)$ holds;
otherwise $\Minprocess(v,x,Q(v)) = w$,
where
$w = \min\{u.x \mid u \in S_v,
v.\invar{\groupdist}[u.x] = u.\invar{\groupdist}[u.x]
\allowbreak +1\}$.
The action
``$v.x \longleftarrow \Minprocess(v,x,Q(v))$''
lets
process $v$ store on $v.x$
the minimum identifier of the processes
satisfying $Q$ in its group;
$v$ stores $\bot$ on $v.x$
if no process of the group satisfies $Q$.
For example, for every $v \in V$ and $u \in N_v^{k+1}$,
action $\refm{border}$ in Table \ref{tbl:merge}
lets $v$ store on $v.\border[u]$
the process with the minimum identifier
in group $l_v$ which
neighbors to a process in group $u$,
if such a process exists.
For $u,v \in V$, array variable $x$, and set $X \subseteq N_v$,
we define $\Distance(v,u,x,X) = 0$ if $u =v$;
otherwise $\Distance(v,u,x,X) = 1 + \min_{w \in X} w.x[u]$. 
By an abuse of notation, we define $1 + \bot = \bot$
throughout the paper.
Hence, $\Distance(v,u,x,X)=\bot$
if $v \neq u$ and $X = \emptyset$.
This macro is useful to compute
some kind of distance between $u$ and $v$.
For example, 
for every $v \in V$ and $u \in S_v$,
$v$ eventually has
$v.\invar{\groupdist}[u] = d_{S_v}(u,v)$
by action $\refi{groupdist}$
in Table \ref{tbl:init}.

\subsection{Algorithm $\init$}
Algorithm $\init$ creates an initial partition based on
the BFS-tree, and computes the values of the variables used in $\merge$,
such as $\domain$ and $\dist$.
The functions used to describe $\init$ in Table \ref{tbl:init}
are defined in Table \ref{tbl:initfunctions}.
Algorithm $\init$ first computes and stores
$N_v^{\dmax+1}$ on $v.\domain$
and 
stores $d(v,u)$ on $v.\dist[u]$ for each $u \in N_v^{\dmax+1}$
within $O(\dmax)$ rounds ($\refi{domain}$ and $\refi{dist}$).
Next, $\init$ gives
an initial partition based on heights of processes in the BFS-tree
and stores the group identifier on each $v.\initgroup$
($\refi{height}$ and  $\refi{initgroup}$).
After that, $\init$ initializes
five copying variables $\invar{\group}$,
$\invar{\groups}$, $\invar{\groupdist}$
$\invar{\stampon}$, and $\invar{\prior}$
($\refi{group}$, $\refi{groups}$,
$\refi{groupdist}$, $\refi{stampon}$, and $\refi{prior}$).
The number of initial groups is at most $2n/k + 1$
since every initial group has at least $k/2$ processes,
except for the group including the root of the BFS-tree.

We define error-detecting variable $\errmerge(v)$ in Table \ref{tbl:error}.
(Ignore $\Stampok(v,u)$ before the middle of Section \ref{sec:merge}.)
The following two lemmas hold trivially
by the definition of $\errorM = \bigvee_{v \in V}\errmerge(v)$
(in spite of the definition of $\Stampok(v,u)$).
\begin{lemma}
 \label{lem:init}
 Every maximal execution of $\init$
 terminates at a configuration
 in $\lnot \errorM$ within $O(\dmax)$ rounds.
 \end{lemma}
 \begin{lemma}
 \label{lem:numgroups}
 The number of groups,
 \ie $|\{l_v \mid v \in V\}|$,
 is no more than $2n/k + 1$
 at a configuration in $\lnot \errorM$.
 \end{lemma}
       
\subsection{Algorithm $\merge$}
\label{sec:merge}
Assuming the network remains at
configurations in $\lnot \errorM$,
$\merge$ merges two near groups
if mergeable;
Otherwise, it stamps
them 
to memorize that they are unmergeable
by using three variables
$\stampone$, $\stamptwo$, and $\stampdist$.
The functions used to describe $\merge$ in Table \ref{tbl:merge}
are defined in Table \ref{tbl:mergefunctions}.

 A group $g(i)$ is called $g(j)$'s {\em candidate\/} if
 $g(i)$ is near to $g(j)$
 and $j.\invar{\stampon}(i)$ does not hold,
 and $g(i)$ is said to be \emph{prior}
 if $i.\invar{\prior}[i]$ holds;
 otherwise $g(i)$ is \emph{non-prior}.
 Note that, if $g(i)$ is near to $g(j)$,
 all nodes $v \in g(j)$
 agree on their
 $\invar{\stampon}[i]$ and $\invar{\prior}[i]$,
 and $v.\invar{\prior}[i] = i.\invar{\prior}[i]$
 since we assume that the network stays in $\lnot \errorM$.
%
 
 In actions $\refm{border}, \dots, \refm{stamptwo}$,
 each group chooses one of its candidates
 and checks mergeability with it.
 Specifically, all processes of each group $g(i)$
 agree to choose the same group $g(j)$ as their target;
 $g(j)$ is the prior candidate that has the minimum $j$
 if such a candidate exists;
 otherwise $g(j)$ is the candidate that has the minimum $j$
 ($\refm{border}$, $\refm{far}$, and $\refm{target}$).
 After actions
 $\refm{border},\dots,\refm{target}$ converge\footnotemark{},
 either $\Detector(j,i)$ or $\Detector(i,j)$ holds
 since $j = \Target(i)$. In what follows,
 we assume $\Detector(j,i)$ without loss of generality.
 \footnotetext{We say that actions
 $A_1,A_2,\dots,A_s$ \emph{converge}
 if none of them is enabled at any process.}
 Then, action $\refm{mergedist}$ leads to $v.\mergedist[u] = d_{g(i)\cup g(j)}(v,u)$
 for all $v \in g(j)$ and $u \in g(i)$ (Lemma \ref{lem:mergedist}).
 Let $\vfind = \min\{v \in g(j) \mid \exists u \in g(i): d_{g(i)\cup g(j)}(v,u) = \dmax + 1\}$ and
 $\vfound = \min\{u \in g(i) \mid d_{g(i)\cup g(j)}(u,\vfind)=k+1\}$.
 If $g(i)$ and $g(j)$ are unmergeable,
 all processes $v \in g(j)$ compute $\vfind$ using $\mergedist$,
 and make the stamp such that
 $(v.\stampone[i],v.\stamptwo[i],v.\stampdist[i])=
 (\vfind,\bot,d_{g(i) \cup g(j)}(v,\vfind))$,
 and all processes $u \in g(i)$
 make the stamp such that
 $(u.\stampone[j],u.\stamptwo[j],u.\stampdist[j])=(\bot,\vfound,d_{g(i)
 \cup g(j)}(u,\vfind))$
 ($\refm{stampone},\dots,\refm{stamptwo}$).
 If $g(i)$ and $g(j)$ are mergeable,
 then $\vfind = \bot$,
 thus all processes $v \in g(j)$ and $u \in g(i)$
 store $\bot$ on $v.\stampdist[i]$ and $u.\stampdist[j]$
 ($\refm{stampone}$ and $\refm{stampdist}$).
 Old stamps between two groups $g(i')$ and $g(j')$
 remain when $\Target(i') \neq j'$ and $\Target(j') \neq i'$
 hold ($\refm{stampone},\dots,\refm{stamptwo}$).

 After actions $\refm{border},\dots,\refm{stamptwo}$ converge, 
 two groups $g(i)$ and $g(j)$ merge by $\refm{group}$
 (\ie $v.\group = \min\{i,j\}$ for all $v \in g(i)\cup g(j)$)
 if and only if
 the two groups are mergeable
 and they target each other,
 which can be easily determined from the values of $\target$ and $\stampdist$.
 When a group is not merged,
 all processes $v$ of the group remain at the same group
 (\ie we have $v.\group = v.\invar{\group}$).
 Thereafter,
 $\groups$ and $\groupdist$
 are updated by $\refm{groups}$ and $\refm{groupdist}$.
 The variable $\stampon$ represents the validity
 of the stamp.
 If a group merges with its target,
 the stamps regarding the group should be removed.
 Hence,
 $v.\stampon[u]$ is set to $\fl$
 not only when $v.\stampdist[u] = \bot$
 but also when 
 $v$ and/or $u$ merges with some group;
 otherwise $v.\stampon[u]$ is set to $\tr$
 ($\refm{merging}$ and $\refm{stampon}$).
 Variable $\prior$ is updated
 by the following simple policy ($\refm{prior}$): 
 a non-prior group becomes a prior group
 when it merges with its target,
 and a prior group becomes a non-prior group
 when its $\stampon[i]$ holds
 for every near group $g(i)$.
 Thanks to this policy, the number of iterations
 of $\merge$ is bounded by $O(n/k)$
 (Lemma \ref{lem:iteration}).

 \begin{table*}[h]
 \refstepcounter{tbl}
  \label{tbl:merge}
 \center
 \begin{tabular}{llll}
  \toprule
  \multicolumn{4}{c}
  {\textbf{Table \ref{tbl:merge}: $\merge$}} \\
  \midrule
  \multicolumn{4}{l}{\textbf{[Actions of process $v$]}}\\
  \msubstitute{merge:border}
  {$v.\border[u]$}
  {$\Minprocess(v,\border[u],N_v(u)\neq \emptyset)$}
  \msubstitute{merge:far}
  {$v.\far[u]$}
  {$\Minprocess(v,\far[u], \exists w \in g_v(u): v.\dist[w]=k+1)$}
  \msubstitute{merge:target}
  {$v.\target[u]$}
  {$\Share(v,u,\target,\Target(v))$}
  \msubstitute{merge:mergedist}
  {$v.\mergedist[u]$}
  {$\Mergedist(v,u)$}
  \msubstitute{merge:stampone}
  {$v.\stampone[u]$}
  {$\Stampone(v,u)$}
  \msubstitute{merge:stampdist}
  {$v.\stampdist[u]$}
  {$\Stampdist(v,u)$}
  \msubstitute{merge:stamptwo}
  {$v.\stamptwo[u]$}
  {$\Minprocess(v,\stamptwo[u], \Stampdist(v,u)=k+1)$}
  \msubstitute{merge:group}
  {$v.\group$}
  {$\Group(v)$}
  \msubstitute{merge:groups}
  {$v.\groups[u]$}
  {$\Share(v,u,\groups,v.\group)$}
  \msubstitute{merge:groupdist}
  {$v.\groupdist[u]$}
  {$\Distance(v,u,\groupdist,\{w \in N_v \mid w.\group = v.\group\})$}
  \msubstitute{merge:merging}
  {$v.\merging[u]$}
  {$\Share(v,u,\merging,\Merging(v))$}  
  \msubstitute{merge:stampon}
  {$v.\stampon[u]$}
  {$v.\stampdist[u]\neq \bot
  \wedge \lnot \Merging(v)
  \wedge \lnot v.\merging[u]$}
  \msubstitute{merge:prior}
  {$v.\prior[u]$}
  {$\Share(v,u,\prior,\Prior(v))$}
  \bottomrule
 \end{tabular}
 \vspace{-0.2cm}
\end{table*}

 \begin{table*}[t]
 \center
 \refstepcounter{tbl}
 \label{tbl:mergefunctions}
  \begin{tabular}{c}
   \toprule
       \textbf{Table \ref{tbl:mergefunctions}: Functions used to describe
   $\merge$}\\
   \midrule
 ${\begin{aligned}
 \Candidates(v) &=
 \{u \in v.\domain\setminus \{l_v\}: v.\border[u] \neq \bot
 \wedge v.\far[u] = \bot
 \wedge \lnot \ v.\invar{\stampon}[u]\}\\
 \Target(v) &=
 \begin{cases}
  \min \{u \in \Candidates(v) \mid v.\invar{\prior}[u]\}
  &\textif \exists u \in \Candidates(v): v.\invar{\prior}[u]\\
  \min \Candidates(v)
  &\otherwise
 \end{cases}\\
 \Mergedist(v,u) &= 
 \begin{cases}
  \Distance(v,u,\mergedist,S_v\cup N_v(\Target(v)))
  & \textif u \in g_v(l_v)\\
  \Distance(v,u,\mergedist,S_v\cup N_v(v.\invar{\groups}[u]))
  & \otherwise
 \end{cases}\\
 \Detector(v,u) &\equiv
 l_v = v.\target[u]
 \wedge (l_v \le u \vee \Target(v) \neq u)\\
 \Stampone(v,u) &=
 \begin{cases}
  \Minprocess(v,\stampone[u],
  \exists w \in g_v(u):v.\mergedist[w] = k+1)
  &  \textif \Detector(v,u)\\
  v.\invar{\stampone}[u]
  & \textif v.\invar{\stampon}[u]\\
  \bot & \otherwise
 \end{cases}\\
 \Stampdist(v,u) &=
 \begin{cases}
  0 & \textif v = \Stampone(v,u)\\
  1+\min\left(
  \begin{aligned}
   &\{w.\stampdist[u]\mid w \in S_v\}\\
   &\cup\{w.\stampdist[l_v]\mid w \in N_v(u)\}
  \end{aligned}
  \right)
  & \otherwise
 \end{cases}\\
 \Merging(v) &\equiv
 l_v = v.\target[\Target(v)]
 \wedge v.\stampdist[\Target(v)]=\bot\\ 
 \Group(v) &=
 \begin{cases}
  \min\{g_v(v),\Target(v)\}&\textif \Merging(v)\\
  g_v(v) & \otherwise
 \end{cases}\\
 \Saturated(v)
 &\equiv \forall u \in v.\domain: (v.\border[u] \neq \bot
 \wedge v.\far[u] = \bot)
 \Rightarrow v.\stampon[u]\\
 \Prior(v) &\equiv \Merging(v) \vee (v.\invar{\prior}[v]
 \wedge \lnot \Saturated(v))
   \end{aligned}}$\\
       \bottomrule
  \end{tabular}
 \end{table*}
 
Now, let us see the definition of $\Stampok(v,u)$ in Table \ref{tbl:error}.
Suppose that an execution of $\merge$ terminates
at some configuration $\gamma$.
Then, it is easily shown that
$v.\invar{\stampon}[u] \Rightarrow \Stampok(v,u)$
holds for all $v \in V$ and $u \in v.\domain$
in $\gamma^{\rmcopy}$.
(See Lemma \ref{lem:stampok}.)

\begin{lemma}
 \label{lem:stampon}
 Consider two groups $g(i)$ and $g(j)$ are near
 in a configuration $\gamma \in \lnot \errorM$.
 Then, $g(i)$ and $g(j)$ are not mergeable
 if $v.\invar{\stampon}[j]$ holds for some $v \in g(i)$.
\end{lemma}
\begin{proof}
 Since $\gamma \in \lnot \errorM$ holds,
 (i) all processes $w_1 \in g(i)$ have a common non-null
 $\vfind = w_1.\invar{\stampone}[j] \in g(i)$
 and all processes $w_2 \in g(j)$ have
 a common non-null
 $\vfound = w_2.\invar{\stamptwo}[i] \in g(j)$,
 or (ii)
 all processes $w_1 \in g(i)$ have a common non-null
 $\vfound = w_1.\invar{\stamptwo}[j] \in g(i)$
 and all processes $w_2 \in g(j)$ have
 a common non-null
 $\vfind = w_2.\invar{\stampone}[i] \allowbreak \in g(j)$.
 In both cases, $\lnot \errorM$ guarantees
 $d_{g(i)\cup g(j)}(\vfind,\vfound)=k+1$,
 which means $\diam(\gsub(g(i)\cup g(j)))\ge k+1$.
\end{proof}
 
\begin{lemma}
\label{lem:mergedist}
 Let $\gamma$ be a configuration in $\lnot \errorM$
 where $\refm{border}, \dots, \refm{mergedist}$ converges
 and $j = \Target(i)$.
 In configuration $\gamma$,
 we have $v.\mergedist[u] = d_{g(i)\cup g(j)}(v,u)$
 for all $v \in g(j)$ and $u \in g(i)$.
\end{lemma}

\begin{proof}
 We let $d(v) = d_{g(i)\cup g(j)}(v,u)$ 
 for any fixed $u \in g(i)$.
 Clearly, $u.\mergedist[u]=d(u)=0$ holds in $\gamma$.
 Let $w$ be a process of $g(j)\cup g(i) \setminus \{u\}$.
 We have $w.\mergedist[u] \ge d(w)$
 since $w.\mergedist[u] =
 1+\min\{w'.\mergedist[u] \mid w' \in N_w(i)\cup N_w(j)\}$
 and $w''.\mergedist[u] = 0$ only if $w''\neq u$.
 On the other hand,
 since there exists a path $w,w_1,w_2,\dots,u$ of length $d(w)$ in
 $G(g(i)\cup g(u))$,
 we have
 $w.\mergedist[u] \le 1+w_1.\mergedist[u] \le 2 + w_2.\mergedist[u] \le
 \dots \le d(w) + u.\mergedist[u] = d(w)$.
 \end{proof}

Let $x_i$ be the variable updated by action $\rmM_i$,
and $\calG_i(v)$ be the guard for action $\rmM_i$.
For example, variable $x_{\ref{merge:group}}$ denotes $\group$
and guard $\calG_{\ref{merge:prior}}(v)$ denotes
$\exists u \in v.\domain:v.\prior[u] \neq \Share(v,u,\prior,\Prior(v))$.
\begin{lemma}[Shiftable Convergence]
 \label{lem:shiftable}
 Every maximal execution $\varrho = \gamma_0, \gamma_1, \dots$
 of $\merge$ where $\gamma_0 \in \lnot \errorM$
 terminates at configuration $\gamma$
 such that $\gamma^{\rmcopy} \in \lnot \errorM$,
 within $O(\dmax)$ rounds.
\end{lemma}

\begin{proof}
 It is guaranteed by Lemma \ref{lem:mergeterminate},
 as shown below,
 that $\varrho$ terminates
 within $O(\dmax)$ rounds.
 Hence, it suffices to show
 $\gamma^{\rmcopy} \in \lnot \errorM$
 where $\gamma$ is the configuration at which $\varrho$ terminates.
 Recall that $\errorM \equiv \bigvee_{v \in V} \errmerge(v)$,
 and $\errmerge(v)$ is defined as follows;
 Since $\gamma_0 \in \lnot \errorM$
 and execution $\varrho$ never updates
 variables $\domain$, $\dist$, $\height$, and $\initgroup$,
 configuration $\gamma^{\rmcopy}$ satisfies
 $(v.\domain =
\Domain(v)) \wedge (\forall u \in v.\domain:v.\dist[u] = \Dist(v,u))
 \wedge (v.\height = \Height(v)) \allowbreak \wedge (v.\initgroup = \Initgroup(v))$
 for any $v \in V$.
 Since $\varrho$ never
 assigns different groups to distinct processes in
 the same group at $\gamma_0$,
 $\gamma^{\rmcopy}$ satisfies $\Groupok(v)$ for any $v \in V$.
 By Lemmas \ref{lem:groupsok}, \ref{lem:groupdistok},
 and \ref{lem:stampok} shown below,
 in configuration $\gamma^{\rmcopy}$, we have
 $\Groupsok(v)$, $\Groupdistok(v)$,
 and $\forall u \in v.\domain: \allowbreak v.\stampon[u] \Rightarrow \Stampok(v,u)$
 for any $v \in V$.
 We also have $\forall u \in N_v,
 \forall w \in v.\domain \cap u.\domain:\allowbreak
 v.\invar{\prior}[w]=u.\invar{\prior}[w]$
 for any $v \in V$ in $\gamma^{\rmcopy}$,
 since $\gamma$ satisfies $v.\prior[u] = \Share(v,u,\prior, \Prior(v))$
 for all $v \in V$ and $u \in v.\domain$.
\end{proof}

\begin{lemma}
 \label{lem:mergeterminate}
 Every maximal execution $\varrho = \gamma_0, \gamma_1, \dots$
 of $\merge$ where $\gamma_0 \in \lnot \errorM$
 terminates within $O(\dmax)$ rounds.
\end{lemma}
\begin{proof}
 Guard $\calG_i(v)$ is independent from the value of $u.x_j$
 for any $u \in N_v^{1}$ and $j > i$.
 Hence, in an execution of $\merge$,
 $v.x_j$ is never updated after actions $M_0,\dots,M_j$ converged
 (\ie $\forall u \in V, i \in [0,j]: G_i(u) = \fl$).
 Therefore, by the description of $\merge$,
 any action $M_j$ converges within $O(\dmax)$ rounds
 after actions $M_0,\dots,M_{j-1}$.
 This proves the lemma.
\end{proof}

\begin{lemma}
 \label{lem:groupsok}
 If a maximal execution $\varrho = \gamma_0, \gamma_1, \dots$
 of $\merge$ where $\gamma_0 \in \lnot \errorM$
 terminates at configuration $\gamma$,
 then $\gamma^{\rmcopy}$ satisfies
 $\Groupsok(v)$ for all $v \in V$.
\end{lemma}
\begin{proof}
 Recall the definition of  $\Groupsok(v)$ in Table \ref{tbl:error}.
 By that definition, it suffices to show that,
 in configuration $\gamma$, we have
 $v.\groups[u] = \Share(v,u,\groups,v.\group)$ for all $u \in v.\domain$
 and $v.\groups[v.\group] = v.\group$.
 The former condition holds since $\calG_{\ref{merge:groups}}$
 holds in $\gamma$.
 The latter condition holds
 if $g(l_v)$ is not merged in execution $\varrho$,
 since $v.\group = l.\group = l$ holds in $\gamma$.
 Even if $g(l_v)$ merges with $g(u)$
 in $\varrho$,
 the latter condition holds
 since $\calG_{\ref{merge:group}}$ guarantees 
 $v.\group = \min\{l_v,u\}.\group = \min\{l_v,u\}$ in $\gamma$.
\end{proof}

\begin{lemma}
 \label{lem:groupdistok}
 If a maximal execution $\varrho = \gamma_0, \gamma_1, \dots$
 of $\merge$ where $\gamma_0 \in \lnot \errorM$
 terminates at configuration $\gamma$,
 then $\gamma^{\rmcopy}$ satisfies
 $\Groupdistok(v)$ for all $v \in V$.
\end{lemma}
\begin{proof}
 Recall the definition of  $\Groupdistok(v)$ in Table \ref{tbl:error}.
 By that definition,
 it suffices to show that, in configuration $\gamma$,
 we have 
 $v.\groupdist[u] =
 \Distance(v,u,\{w.\groupdist[u] \mid w \in N_v, w.\group = v.\group\})$
 and $v.\groupdist[u] \neq k+1$
 for all $u \in \{w \in N_v \mid w.\group = v.\group\}$.
 The former condition holds since $\calG_{\ref{merge:groupdist}}$
 holds in $\gamma$.
 The latter condition holds
 if $g(l_v)$ is not merged in execution $\varrho$,
 since $v.\groupdist[u] = d_{g(l_v)}(v,u) \in [0,k]$
 holds 
 in $\gamma$.
 If $g(l_v)$ merges with some group $g(w)$ in $\varrho$,
 then $d_{g(l_v)\cup g(w)}(v,u) \in [0,k]$ holds;
 otherwise $g(l_v)$ cannot be merged with $g(w)$.
 Hence, even if $g(l_v)$ merges with $g(w)$
 in $\varrho$, the latter condition holds since
 $\calG_{\ref{merge:groupdist}}$
 guarantees
 $v.\groupdist[u] = d_{g(l_v)\cup g(w)}(v,u) \in [0,k]$ in
 $\gamma$.
 \end{proof}

\begin{lemma}
 \label{lem:stampok}
 If a maximal execution $\varrho = \gamma_0, \gamma_1, \dots$
 of $\merge$ where $\gamma_0 \in \lnot \errorM$
 terminates at configuration $\gamma$,
 then $\gamma^{\rmcopy}$ satisfies
 $v.\invar{\stampon}[u]\Rightarrow \Stampok(v,u)$
 for all $v \in V$ and $u \in v.\domain$.
\end{lemma}
 \begin{proof}
 Recall the definition of  $\Stampok(v)$ in
 Table \ref{tbl:error}.
 Assuming that $v.\invar{\stampon}[u]$ holds
 in $\gamma^{\rmcopy}$,
 we will prove that $\Stampok(v,u)$ holds in $\gamma^{\rmcopy}$.
 Only in this proof,
 we denote $w_1.\stampone[u]$, $w_1.\stamptwo[u]$,
 and $w_1.\stampdist[u]$ by just
 $w_1.\stampone$, $w_1.\stamptwo$,
 \allowbreak and $w_1.\stampdist$ for all $w_1 \in g(l_v)$,
 and we denote $w_2.\stampone[l_v]$, $w_2.\stamptwo[l_v]$,
 and $w_2.\stampdist[l_v]$ by just $w_2.\stampone$,
 $w_2.\stamptwo$,
 and $w_2.\stampdist$ for all $w_2 \in g(u)$.
 Since $v.\stampon[u]$ holds in $\gamma$,
 either $g(l_v)$ or $g(u)$ never merges in execution $\varrho$,
 and  there exists $\vfind \in g(l_v) \cup g(u)$
 such that $w.\stampdist = d_{g(l_v)\cup g(u)}(w,\vfind)$
 holds for all $w \in g(l_v)\cup g(u)$.
 Let $\vfound = \min\{w \in g(l_v)\cup g(u)
 \mid d_{g(l_v)\cup g(u)}(w,\vfind)=k+1\}$.
 (Note that $\vfound \neq \bot$.)
 In $\gamma$, we have $(w.\stampone,w.\stamptwo) = (\vfind,\bot)$
 for all $w \in g(\vfind.\invar{\group})$,
 and $(w.\stampone,w.\stamptwo) = (\bot,\vfound)$
 for all $w \in g(\vfound.\invar{\group})$.
 All of these prove that $\Stampok(v,u)$ holds.
 \end{proof}

 In the rest of this section,
 we denote $\cgoal(\merge,\errmerge)$ simply by $\cgoal$.

\begin{lemma}[Correctness]
 \label{lem:correctness}
 Every $\gamma \in \cgoal$
 satisfies $\leg$.
\end{lemma}

\begin{proof}
 Let $V(i) = \{v \in V \mid v.\group = i\}$.
 In $\gamma \in \cgoal$,
 we have (i) $\diam(g(i)) \le k$ for all $i \in V \ \st\ g(i) \neq \emptyset$,
 (ii) $\Candidates(v)=\emptyset$ for all $v \in V$,
 and (iii) $g(i) = V(i)$ for all $i \in V$.
 By Lemma \ref{lem:stampon} and (ii), 
 every $g(i)$ is not mergeable with any other group in $\gamma$.
 Hence, by (i) and (iii),
 $\{V(i) \neq \emptyset \mid i \in V\}$
 is a minimal partition of $G$
 with diameter-bound $\dmax$ in $\gamma$.
 \end{proof}

\begin{lemma}[Loop Convergence]
 \label{lem:iteration}
 Let $\varrho_0, \varrho_1 \dots$ 
 be an infinite sequence of maximal executions of $\merge$
 where $\varrho_i =\gamma_{i,0},\gamma_{i,1},\dots,\gamma_{i,s_i}$,
 $\gamma_{0,0} \in \lnot \errorM$,
 and $\gamma_{i+1,0} = \gamma_{i,s_i}^\rmcopy$ for $i \ge 0$.
 Then, $\gamma_{j,s_j} \in \cgoal$
 and $\round(\varrho_0)+\round(\varrho_1)+\dots\round(\varrho_j) =O(n)$
 hold for some $j = O(n/\dmax)$.
\end{lemma}

\begin{proof}
 We say that a group $g(v)$ is black
 if $g(v)$ is non-prior and
 some non-prior group $g(u)$ exists
 such that $u \in \Candidates(v)$;
 otherwise $g(v)$ is white.
 Note that a black group may become white,
 but no white group becomes black.
 We denote the number of groups (\ie $|\{l_v \mid v \in V\}|$),
 the number of black groups,
 and the number of prior groups
 in configuration $\gamma \in \lnot \errorM$
 by $\numbergroup(\gamma)$, $\numberblack(\gamma)$,
 and $\numberprior(\gamma)$ respectively.
 We define $\total(\gamma)  = 2\numbergroup(\gamma)+\numberprior(\gamma)+\numberblack(\gamma)$.
 Note that
 both $2\numbergroup(\gamma)+\numberprior(\gamma)$ and $\numberblack(\gamma)$ are
 monotonically non-increasing,
 thus $\total(\gamma)$ is also
 monotonically non-increasing.
 The former $2\numbergroup(\gamma)+\numberprior(\gamma)$ decreases every time any two groups merges.
 In what follows, we show
 $\total(\gamma_{i,0}) > \total(\gamma_{i+2,0})$
 if $\gamma_{i} \notin \cgoal$.
 Consider first the case that
 there exists at least one
 prior group in $\gamma_{i,0}$.
 Let $g(j)$ be the group with the minimum $j$
 among those prior groups. 
 If $g(j)$ is also the minimum prior group
 in $\gamma_{i+1,0}$,
 then $g(j)$ merges with some group
 or becomes saturated and non-prior
 in $\varrho_i$ or $\varrho_{i+1}$,
 which gives $\total(\gamma_{i,0}) > \total(\gamma_{i+2,0})$.
 If $g(j)$ is no longer the minimum prior group
 in $\gamma_{i+1,0}$,
 some group must merge and become prior
 in $\varrho_i$,
 hence $\total(\gamma_{i,0}) > \total(\gamma_{i+2,0})$ holds.
 Consider next the case that there exists
 no prior group in $\gamma_{i,0}$.
 Let $g(j')$ be the group with the minimum $j'$
 such that $\Candidates(j') \neq \emptyset$
 in $\gamma_{i,0}$.
 By a similar discussion,
 in $\varrho_i$ or $\varrho_{i+1}$,
 $g(j)$ becomes saturated and non-prior,
 or some pair of groups merges. 
 In the former case,
 $\numberblack(\gamma_{i,0}) \le \numberblack(\gamma_{i+2,0})$
 holds.
 In the latter case,
 $2\numbergroup(\gamma_{i,0})+\numberprior(\gamma_{i,0}) \le 2\numbergroup(\gamma_{i+2,0})+\numberprior(\gamma_{i+2,0})$ holds.
Thus, in any case, 
 $\total(\gamma_{i,0}) > \total(\gamma_{i+2,0})$
 holds if $\gamma_{i} \notin \cgoal$.
 By Lemma \ref{lem:numgroups},
 $\total(\gamma) \le 4 \numbergroup(\gamma) = O(n/\dmax)$ holds.
 Hence, we have $\gamma_{j,s_j} \in \cgoal$ for some $j = O(n/\dmax)$,
from which $\round(\varrho_0)+\round(\varrho_1)+\dots\round(\varrho_j) =O(n)$ follows
because $\round(\varrho_i) = O(k)$ holds for each $i$ by Lemma \ref{lem:shiftable}.
\end{proof}

\begin{theorem}
 \label{therem:mei}
 Algorithm $\mei$ is silent self-stabilizing for $\leg$.
 Every maximal execution of the algorithm
 terminates within $O(n\diam /\dmax)$ rounds.
 The number of groups it produces,
 \ie $|\{v.\group \mid v \in V\}|$,
 is at most $2n/k + 1$.
 The space complexity per process of $\mei$
 is $O((n + \nfalse)\log n)$
 where $\nfalse$ is the number of identifiers
 stored in $v.\domain$ for some $v \in V$
 in an initial configuration,
 and which do not match the identifier of any $u \in V$.
\end{theorem}

\begin{proof}
 The first three arguments are proven
 by Lemmas \ref{theorem:aep}, \ref{lem:init},
 \ref{lem:numgroups},
 \ref{lem:shiftable}, \ref{lem:correctness},
 and \ref{lem:iteration}.
 The last argument about the space complexity
 is trivial.
\end{proof}

\section{Conclusion}
We have given
a silent self-stabilizing algorithm
for the minimal $\dmax$-grouping problem.
Given a network $G$ and a diameter-bound $\dmax$,
it guarantees that,
regardless of an initial configuration,
the network reaches a configuration
where
the diameter of every group is no more than $\dmax$
and
no two groups can be merged without violating the diameter-bound.
Its time complexity is $O(n\diam / \dmax)$
and its space complexity per process
is $O((n + \nfalse)\log n)$.
The number of groups it produces
is at most $2n/k + 1$.
A novel composition technique called
\emph{loop composition} is also given
and used in our algorithm.
 
\bibliographystyle{plain}
\bibliography{group}

\end{document}